\documentclass[aps,prl,showpacs,notitlepage,twocolumn,superscriptaddress,nofootinbib,preprintnumbers]{revtex4-2}
\usepackage{bbm}
\usepackage{mathrsfs}
\usepackage{epsfig}
\usepackage{soul,xcolor}
\usepackage{graphicx}
\usepackage{amsfonts}
\usepackage{amsthm}
\usepackage[figuresright]{rotating}
\usepackage{amssymb}
\usepackage{amsmath}
\usepackage{dcolumn}
\usepackage{physics}
\usepackage{bm}
\usepackage{verbatim}
\usepackage{braket}
\usepackage[normalem]{ulem}
\usepackage[ruled,vlined,linesnumbered]{algorithm2e}
\usepackage{setspace}
\usepackage[colorlinks,linkcolor=blue,anchorcolor=blue,citecolor=blue,urlcolor=blue]{hyperref}
\newtheorem{theorem}{Theorem}

\begin{document}

\title{Flow-based Phase-space Tomography of Continuous-variable Quantum States}

\author{Owen Dugan}
\affiliation{Department of Physics, Massachusetts Institute of Technology}
\affiliation{Department of Computer Science, Stanford University, Stanford, CA, USA}

\author{Rumen Dangovski}
\affiliation{Department of Physics, Massachusetts Institute of Technology}

\author{Peter Y. Lu}
\affiliation{Department of Physics, Massachusetts Institute of Technology}
\affiliation{Department of Electrical and Computer Engineering, Tufts University, Medford, MA}

\author{Di Luo}
\thanks{diluo@tsinghua.edu.cn}
\affiliation{Department of Physics, Massachusetts Institute of Technology}
\affiliation{Department of Electrical and Computer Engineering, University of California, Los Angeles, CA 90095, USA}
\affiliation{Department of Physics, Tsinghua University, Beijing 100084, China}

\date{\today}

\begin{abstract}
Continuous-variable quantum state tomography is limited by the cost of resolving non-Gaussian structure in high-dimensional phase space. We introduce \emph{QST-Flow}, a quantum state tomography framework via flow-based generative modeling that represents experimentally accessible phase-space quasiprobability distributions with normalized, samplable neural densities rather than a truncated density matrix. The framework has two variants: \emph{QST-QFlow} models the positive Husimi-$Q$ function with a single normalizing flow, while \emph{QST-WFlow} models sign-changing Wigner functions as a trainable difference of two normalized flows. This construction preserves quasiprobability normalization and enables exact density evaluation, direct sampling, and importance-sampled learning from finite phase-space measurements without a fixed grid. Benchmarks on non-Gaussian cat, binomial, Gottesman-Kitaev-Preskill, number, and Fock states show accurate single-mode reconstructions, extension to multimode states, robustness on noisy Wigner data, and improved reconstruction error compared with prior machine-learning tomography methods. QST-Flow opens a promising route toward scalable, measurement-efficient phase-space tomography of nonclassical bosonic systems.
\end{abstract}
\maketitle

\textit{Introduction---.} Continuous-variable (CV) bosonic systems provide a natural language for quantum optics, superconducting cavities, sensing, communication, and bosonic error correction~\cite{lvovsky2009continuous}. Their quantum states are most often interrogated through phase-space quasiprobability distributions such as the Husimi $Q$ function~\cite{husimi1940formal} and the Wigner function~\cite{wigner1932quantum}. These representations are experimentally meaningful: the $Q$ function is associated with coherent-state measurements, while the Wigner function can be obtained from homodyne reconstruction or direct displaced-parity measurements~\cite{smithey1993measurement,banaszek1996direct}. They also expose the central difficulty of CV tomography. The relevant Hilbert space is infinite-dimensional, experimentally available data are finite and noisy, and highly nonclassical states such as cat, binomial, and Gottesman-Kitaev-Preskill (GKP) states have fine phase-space structure that becomes increasingly hard to resolve as the number of modes grows~\cite{mirrahimi2014dynamically,michael2016new,gottesman2001encoding}.

Conventional CV tomography commonly reconstructs a density matrix in a truncated Fock basis or a quasiprobability distribution on a phase-space grid~\cite{lvovsky2009continuous}. These approaches remain powerful for low-dimensional experiments, but their cost is set by the chosen cutoff or grid resolution and grows rapidly with the number of modes. The difficulty is especially acute for Wigner tomography, where interference fringes and negative lobes carry nonclassical information~\cite{kenfack2004negativity} but are vulnerable to coarse sampling, loss, and finite-shot noise. Recent learning-based methods have begun to expand the toolkit for CV quantum tomography. CV classical-shadow approaches adapt randomized measurement ideas to bosonic systems and provide sample-complexity guarantees for selected observables and finite-dimensional shadows~\cite{gandhari2024precision,becker2024classical}. Other learning-theoretic and nonparametric approaches estimate characteristic functions, tomograms, or efficient subclasses of CV states with improved scaling under suitable assumptions~\cite{wu2024efficient,markovich2026nonparametric,mele2025learning}. Despite these advances, learning general non-Gaussian and multimode CV states remains challenging because the model must handle infinite-dimensional phase space, finite and noisy measurements, Wigner negativity, and scalable continuous representations.

\begin{figure*}[t]
    \centering
    \includegraphics[width=0.88\linewidth]{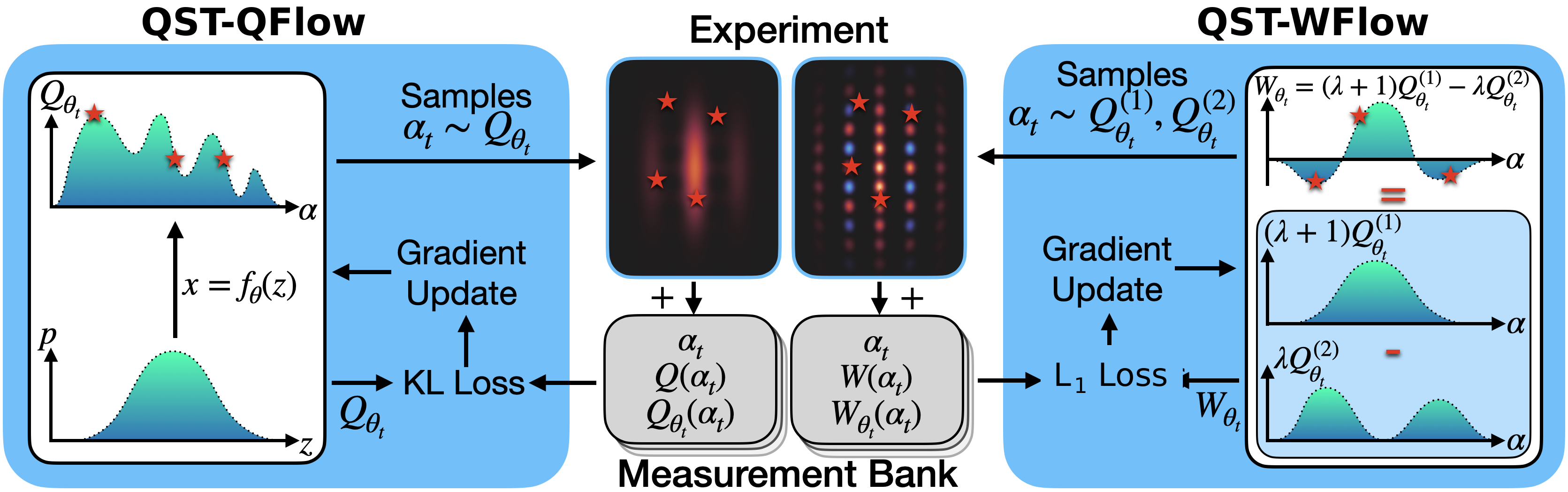}
    \caption{Schematic of QST-QFlow (left) and QST-WFlow (right) tomography for bosonic quantum systems. An experimental state is modeled through a normalizing-flow representation of either the Husimi $Q$ function or the Wigner function. The current flow proposes phase-space points, measurements are acquired at those points, and the model is updated to match the measured phase-space signal.}
    \label{fig:summary}
\end{figure*}

Normalizing flows have recently become an important development in machine learning, with applications across density estimation, simulation, inverse problems, and scientific data analysis. They are especially well suited to the positive part of the CV tomography problem because they map a simple base density to a complex target density through an invertible transformation with a tractable Jacobian. The remaining challenge is to turn this capability into a tomography method that treats both positive $Q$ functions and sign-changing Wigner functions on the same footing. Here we introduce QST-Flow, a unified flow-based generative modeling framework for phase-space tomography of CV quantum states that learns normalized quasiprobability representations directly, without imposing a Fock-basis density-matrix truncation. QST-Flow has two phase-space variants: QST-QFlow, in which the Husimi $Q$ function is represented directly by a normalized flow and trained by sample-based likelihood estimates, and QST-WFlow, in which a Wigner function is represented as a trainable signed combination of two normalized flows that captures negativity while preserving normalization and efficient proposal sampling.

The resulting models provide exact density evaluation, direct sampling, and sample-based inference in continuous phase space. QST-QFlow gives a normalized probability model for Husimi tomography, while QST-WFlow decomposes Wigner functions into the difference of two normalized positive flows, so that negative quasiprobability is modeled explicitly rather than postprocessed. We validate the framework on representative non-Gaussian bosonic states, noisy Wigner data, and multimode states, including a benchmark where QST-WFlow improves on the previous state-of-the-art QST-CGAN~\cite{QST_CGAN} reconstruction. QST-Flow develops an advanced flow-based tomography tool for CV systems, providing a normalized and samplable phase-space representation that reaches beyond single-mode demonstrations and the constraints of conventional grid- or truncation-based methods.

\textit{Phase-space measurements---.} We focus on two experimentally accessible phase-space representations that fully determine a CV quantum state: the Husimi $Q$ function and the Wigner function. The $Q$ function is positive and can therefore be treated as a probability density, while the Wigner function contains the sign changes and negativity that witness nonclassicality~\cite{kenfack2004negativity}. This distinction is the motivation for using a single normalized flow for $Q$ tomography and a signed pair of flows for Wigner tomography. 

First consider the $Q$-function measurement,

\begin{equation}
Q[\rho](\alpha) = \frac{1}{\pi} \text{Tr} \left[ \rho |\alpha\rangle \langle \alpha| \right] = \frac{1}{\pi} \langle \alpha | \rho | \alpha \rangle, 
\end{equation}

which can also be written as

\begin{equation}
Q[\rho](\alpha) = \frac{1}{\pi} \langle 0 | D(-\alpha) \rho D(\alpha) | 0 \rangle
\end{equation}

The \(Q\) function is thus the expectation value of the vacuum projector, \(|0\rangle \langle 0|\), after displacing the field by \(-\alpha\). It is directly measurable, non-negative, normalized, and bounded by $1/\pi$. Physically, this distribution is a Gaussian-smoothed version of the Wigner function and therefore suppresses sub-Planck interference fringes while retaining tomographic information through the overcomplete coherent-state POVM (positive operator-valued measure)~\cite{nielsen2010quantum}. With a finite number of experimental shots, each sampled value has binomial measurement statistics.

Because the $Q$ function corresponds to a coherent-state POVM, the classical square-root overlap of two $Q$ functions bounds the quantum fidelity from above:

\begin{equation}
F(\rho, \sigma) = \left( \text{tr} \sqrt{\sqrt{\rho} \sigma \sqrt{\rho}} \right)^2 \leq \left( \int \sqrt{Q_{\rho} Q_{\sigma}} d \alpha d \alpha^{*} \right)^2
\end{equation}
where $Q_{\sigma}$ denotes the target distribution and $Q_{\rho}$ the flow-based reconstruction. The square-root overlap can be estimated efficiently by sampling from either distribution. We also report the $L_1$ distance $||Q_{\rho}-Q_{\sigma}||$ as a direct phase-space error.

Next consider the Wigner-function measurement. Let \(\mathcal{P}\) be the Hermitian parity operator, which performs reflection about the phase-space origin so that $\mathcal{P} |x\rangle = |-x\rangle, \mathcal{P} |p\rangle = |-p\rangle$. Direct parity-based Wigner measurements have been widely employed because they provide pointwise (local) access to the Wigner function, eliminating the need for the inverse Radon transform required in homodyne tomography~\cite{banaszek1996direct}.

The displaced-parity expression for the Wigner function is
\begin{equation}
W(x, p) = \frac{2}{\pi} \text{Tr}\left[D(-\alpha)\rho D(\alpha)\mathcal{P}\right]
\end{equation}

This expression shows that the Wigner value at $(x,p)$ is obtained by displacing the state by $-\alpha$ and measuring the parity observable $2\mathcal{P}/\pi$. In the Fock basis, the parity operator is $\mathcal{P}=e^{i\pi \hat{a}^\dagger \hat{a}}$, with eigenvalue $+1$ on even photon-number states and $-1$ on odd photon-number states. A finite-shot Wigner measurement is therefore governed by binomial statistics set by the even- and odd-parity weights of $D(-\alpha)\rho D(\alpha)$, and the measured quasiprobability is bounded as $-\frac{2}{\pi} \leq W(x, p) \leq \frac{2}{\pi}$.
The sign of $W(x,p)$ directly reflects this local parity imbalance: negative regions occur when the odd-parity component dominates after displacement, while oscillatory sign changes encode coherence between separated phase-space components. This makes Wigner tomography more sensitive to nonclassical interference than $Q$ tomography, but also more vulnerable to shot noise because a small difference of two parity frequencies must be resolved.

The Wigner-function overlap equals the Hilbert-Schmidt overlap of the corresponding density matrices:

\begin{equation}
    \text{tr} (\rho \sigma) = 2 \pi \int W_{\rho} W_{\sigma} dx dp
\end{equation}
where $W_{\sigma}$ denotes the target and $W_{\rho}$ the flow-based reconstruction. When at least one of the two states is pure, the Hilbert–Schmidt overlap coincides with the squared Uhlmann quantum fidelity. This is the regime used for the state benchmarks below, so the Wigner overlap provides a direct fidelity diagnostic. We also report $||W_{\rho}-W_{\sigma}||_1$ as a direct phase-space error.

\textit{Flow-based phase-space tomography---.} We now describe the two QST-Flow representations, QST-QFlow and QST-WFlow, used for $Q$- and Wigner-function tomography in phase space.

For QST-QFlow, we treat the $Q$ function as a positive probability distribution and develop the flow-based approach of Ref.~\cite{dugan2023q} to quantum state tomography. Normalizing flows transform a simple initial density $p_X$ (typically a unit-normal distribution) to a target density $p_Y$ through an invertible neural map $y=f_\theta(x)$ with $x\sim Q_X$ and $y\sim Q_Y$~\citep{dinh2014nice,rezende2015variational}. The target density is then
\begin{align*}
Q_Y(y) &=Q_X(f_\theta^{-1}(y))\left | \frac{\partial f_\theta^{-1}(y)}{\partial y} \right|.
\end{align*}
This construction gives both normalized probabilities and exact sampling. Possible choices of $f_\theta$ include affine coupling layers (RealNVP)~\citep{dinh2017density}, continuous normalizing flows (CNF)~\citep{grathwohl2018scalable}, and convex potential flows (CP-Flow)~\citep{huang2021convex}. Related universal-approximation results imply that sufficiently expressive flows can represent broad classes of smooth $Q$ functions~\citep{huang2021convex}.

The tomography loop alternates between proposing phase-space points, evaluating the experimental signal at those points, and updating the flow. Training uses an importance-sampled estimate of the KL divergence. Let $B$ be a bank of phase-space samples. Each sample is drawn from a proposal distribution $P$, evaluated with proposal probability $P_\alpha=P(\alpha)$, and stored as $(\alpha,P_\alpha)$. The proposal can be the current flow, a broader initialization distribution, or a mixture of earlier flows; storing $P_\alpha$ keeps the estimator unbiased even as the proposal evolves during training. As shown in the Supplemental Material,
\begin{align*}
    KL(Q_{exp} || Q_{\theta}) &= \int Q_{exp} \ln \frac{Q_{exp}}{Q_{\theta}}\\
    &= \frac{1}{N} \sum_{(\alpha, P_\alpha) \sim B} \frac{Q_{exp}(\alpha)}{P_\alpha} \log \frac{Q_{exp}(\alpha)}{Q_{\theta}(\alpha)},
\end{align*}
where $Q_{\rm exp}$ is the measured $Q$ function, $Q_{\theta}$ is the QST-QFlow model, and $(\alpha,P_\alpha)$ are selected uniformly from the sample bank. This form is useful experimentally because the flow does not require a fixed rectangular grid: samples can be concentrated around peaks, tails, and other regions that dominate either normalization or reconstruction error.

For QST-WFlow, we exploit the fact that any normalized sign-changing phase-space function $W$ can be written as
\begin{equation}
W = (\lambda+1)Q_1-\lambda Q_2
\end{equation}
for normalized non-negative distributions $Q_1$ and $Q_2$. The coefficient $\lambda$ controls the total negative weight: in the exact positive/negative decomposition, $\lambda=\int \max[-W(x,p),0]\,dx\,dp$, and $2\lambda+1=\int |W|$ is the Wigner-function $L_1$ norm. This connects the signed-flow ansatz to the usual volume of Wigner negativity~\cite{kenfack2004negativity}. We therefore parameterize $Q_1$ and $Q_2$ by two normalizing flows and treat $\lambda$ as a trainable parameter constrained to be non-negative. The model is automatically normalized because $(\lambda+1)\int Q_1-\lambda\int Q_2=1$, while sign changes occur wherever the second component dominates the first. To our knowledge, this construction provies the first flow-based generative model to negative quasi-probability distribution beyond the conventional classial machine learning settings.

This decomposition also gives a practical sampling strategy. Although $W$ itself cannot be used as a probability distribution, either $\frac{1}{2}(Q_1+Q_2)$ or a bank of previous such mixtures is a positive proposal with support in the regions where the signed reconstruction can vary. The ansatz is expressive enough to represent Wigner negativity while retaining tractable density evaluation and proposal sampling. It does not, by itself, enforce positivity of the reconstructed density matrix; in this work we evaluate the learned function through experimentally measured phase-space observables and fidelity diagnostics. We train QST-WFlow using an $L_1$ loss estimated by importance sampling,
\begin{align*}
    L_1[W_\theta,W] &= \int |W_\theta - W|\\
    &\approx \frac{1}{N}\sum_{((x,p), P_{(x,p)})\sim B}\frac{|W_\theta(x,p) - W(x,p)|}{P_{(x,p)}},
\end{align*}
where $((x,p),P_{(x,p)})$ are sampled from a proposal bank $B$ (See the Supplementary Material for details).

\begin{figure}[t]
    \centering
    \includegraphics[width=0.95\linewidth]{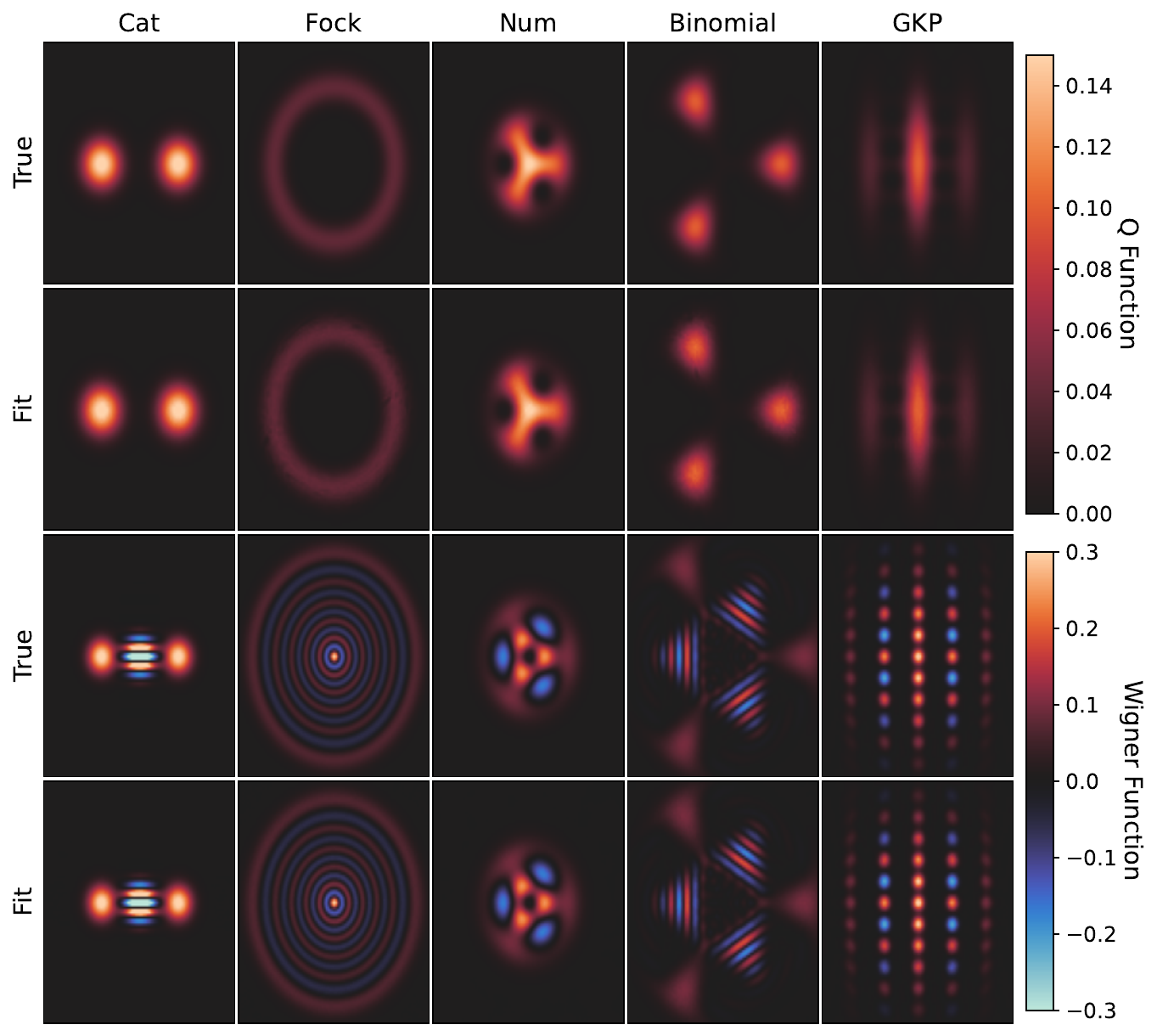}
    \caption{$Q$- and Wigner-function tomography of five single-mode continuous-variable quantum states.}
    \label{fig:1_well}
\end{figure}

\textit{Experiments---.} We first test QST-QFlow and QST-WFlow on single-mode non-Gaussian states, as shown in Fig.~\ref{fig:1_well}. The benchmark set contains the cat state $\ket{\text{cat}(2)}$, the Fock state $\ket{10}_n$, the number-code state $\ket{\text{num}(0)}$, the binomial state $\ket{\text{binom}(5,2)}$, and the GKP state $\ket{\text{GKP}(0.3,20)}$. These examples probe distinct pieces of phase-space physics. Cat states test whether the model resolves interference fringes between separated coherent components~\cite{mirrahimi2014dynamically}. Fock states test rotationally symmetric rings and alternating Wigner signs set by photon-number parity. Binomial and number-code states test sparse superpositions in the Fock basis and are relevant to bosonic error correction~\cite{michael2016new}. GKP states test sharply localized peaks and lattice-scale interference generated by displacement-stabilizer structure~\cite{gottesman2001encoding}. Definitions of the states and architectural details are given in the Supplemental Material. Except for the number of training steps, the same hyperparameters are used across all ten single-mode reconstructions. Most runs use 500 training steps and 5000 samples from the target distribution; the most challenging Wigner reconstructions use up to 4000 steps and at most 40000 samples.

The qualitative difference between the $Q$ and Wigner reconstructions follows the measurement physics above. The $Q$ functions are positive and smoother because coherent-state projection averages over vacuum noise; the corresponding flow mainly needs to learn peak locations, widths, and tails. The Wigner functions contain sharper oscillations and negative lobes, so QST-WFlow must learn cancellations between its two positive components. The signed ansatz is important here: negative regions are not artifacts added after training, but regions where the learned negative component locally outweighs the positive component while the total integral remains fixed.

\begin{figure}[t]
    \centering
    \includegraphics[width=\linewidth]{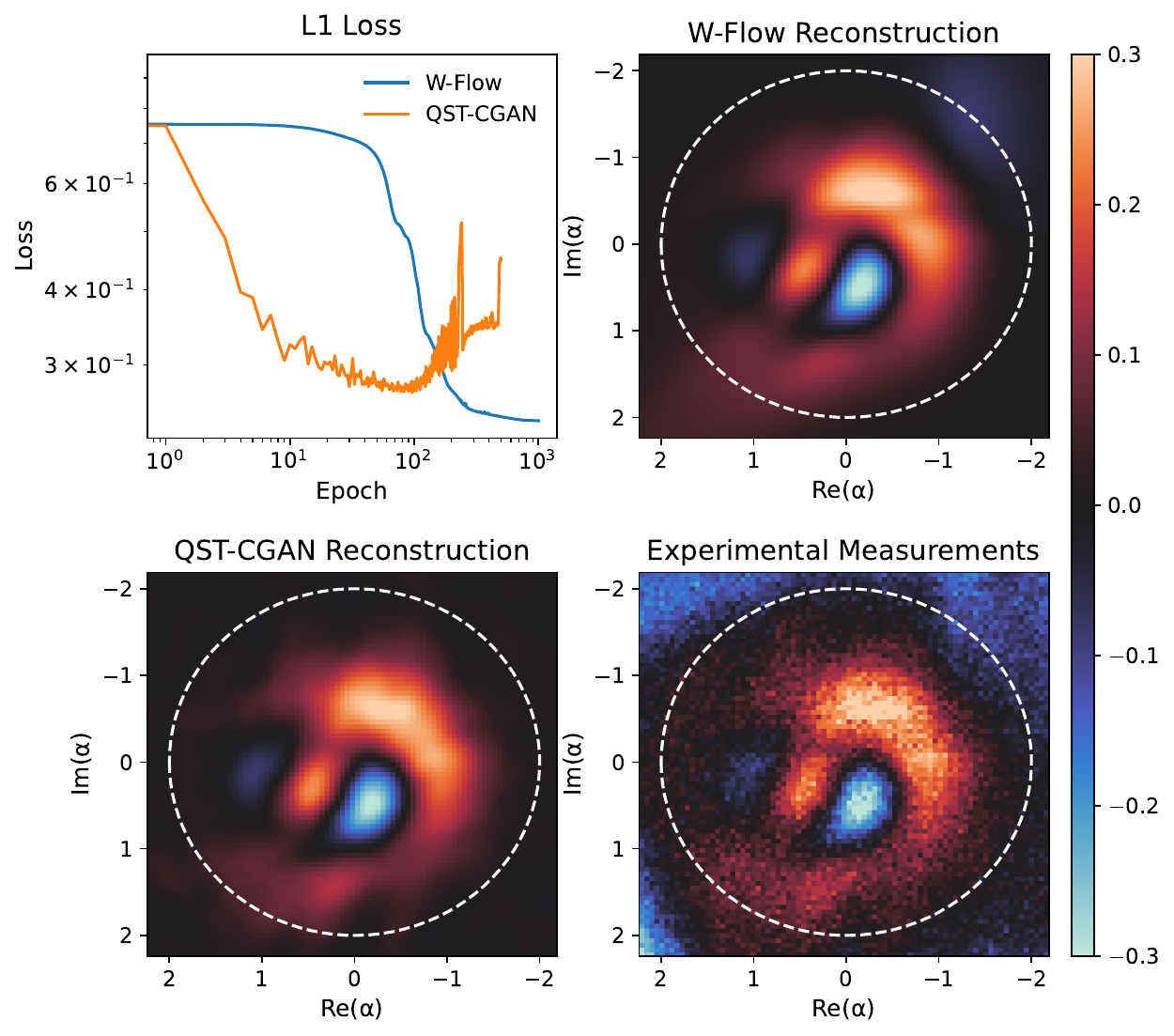}
    \caption{Comparison between QST-WFlow and QST-CGAN on noisy single-mode Wigner-function data.}
    \label{fig:experimental}
\end{figure}

We next compare QST-WFlow with QST-CGAN~\cite{QST_CGAN}, a state-of-the-arts machine learning based CV quantum state tomography approach that outperforms standard maximum-likelihood method. The QST-WFlow model is trained on the same noisy samples used for the QST-CGAN reconstruction. As shown in Fig.~\ref{fig:experimental}, the flow representation gives a lower $L_1$ reconstruction error on this benchmark and more closely follows the measured Wigner-function distribution. The improvement is consistent with the structure of the model: the reconstruction is a normalized signed density rather than an unconstrained image, and the loss is evaluated directly on the measured quasiprobability values. Additional details of the dataset, loss, and training protocol are provided in the Supplemental Material.

Finally, we test scaling beyond a single mode. Multimode CV tomography is difficult for grid-based and density-matrix methods because the phase-space dimension grows linearly with the number of modes while the number of resolved grid points grows exponentially. A flow instead works in the continuous $2M$-dimensional phase space of an $M$-mode system and can draw or score samples without enumerating a tensor-product grid. Figure~\ref{fig:5_well} shows a QST-QFlow reconstruction of the five-mode product state
\[
    \ket{\text{cat}(1.5)}\otimes\ket{2}_n\otimes\ket{\text{num}(0)}\otimes\ket{\text{cat}(0+1i)}\otimes\ket{\text{num}(1)}.
\]
This experiment uses 80000 samples from the target distribution. A ten-mode QST-QFlow reconstruction is shown in the Supplemental Material, together with the corresponding training details. These multimode tests do not remove the fundamental information-theoretic difficulty of reconstructing arbitrary high-dimensional quantum states, but they show that the representation avoids an explicit grid and can exploit sample access to structured CV states.

\begin{figure}[t]
    \centering
    \includegraphics[width=\linewidth]{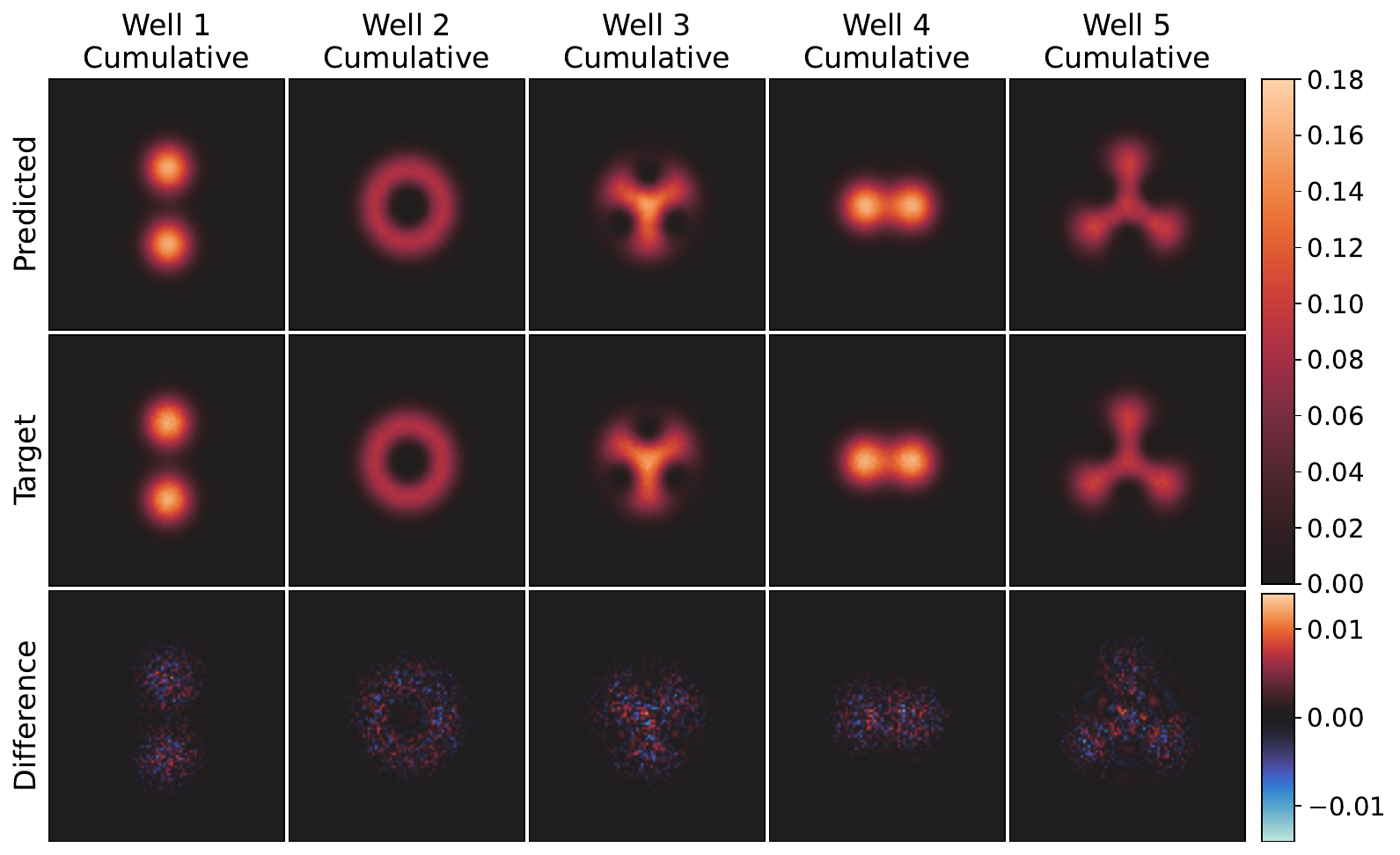}
    \caption{$Q$-function tomography of a five-mode continuous-variable quantum state.}
    \label{fig:5_well}
\end{figure}

\textit{Conclusion---.} We have introduced QST-Flow, a unified flow-based framework for CV quantum state tomography that learns phase-space quasiprobability representations without truncating the density matrix. QST-QFlow treats the Husimi $Q$ function as a positive normalized distribution and uses likelihood-based training with efficient sampling and inference. QST-WFlow treats Wigner tomography as a signed-density learning problem, representing negativity through the difference of two normalized flows while preserving exact quasiprobability normalization and practical proposal sampling. This gives a common continuous phase-space language for smooth positive measurements and highly nonclassical Wigner reconstructions.

The benchmarks show that the method reconstructs representative non-Gaussian bosonic states, extends to multimode states without an explicit tensor-product grid, handles noisy Wigner data, and outperforms the previous QST-CGAN baseline on the noisy Wigner benchmark. Physically, the advantage comes from matching the ansatz to the measurement, where coherent-state $Q$ data are learned as probabilities, while displaced-parity Wigner data are learned as normalized sign-changing quasiprobabilities. More broadly, QST-WFlow generalizes flow-based modeling from ordinary probability densities to normalized quasidistributions, suggesting a route for applying generative modeling ideas to signed and nonclassical data beyond quantum tomography. For future work, useful extensions include enforcing density-matrix physicality within the model, quantifying scaling under realistic loss and detector errors, adding uncertainty estimates for finite-shot parity data, and integrating adaptive sampling with online experimental control.

\textit{Acknowledgments---.} We are grateful to Zhuo Chen for his valuable insights and discussions, which played an important role in the development of this work.DL acknowledges support from Beijing Municipal
Science and Technology Commission and Zhongguancun Science Park Administrative Committee (No.
20251090054).

\bibliographystyle{apsrev4-2}
\bibliography{biblio}

\clearpage
\onecolumngrid

\begin{center}
{\large\bfseries Supplementary Materials on ``Flow-based Phase-space Tomography of Continuous-variable Quantum States''}
\end{center}
\vspace{1em}

\renewcommand\thefigure{S\arabic{figure}}  
\renewcommand\thetable{S\arabic{table}}  
\renewcommand{\theequation}{S\arabic{equation}}
\renewcommand{\theHfigure}{S\arabic{figure}}
\renewcommand{\theHtable}{S\arabic{table}}
\renewcommand{\theHequation}{S\arabic{equation}}
\renewcommand{\thepage}{P\arabic{page}} 
\setcounter{page}{1}
\setcounter{figure}{0}  
\setcounter{table}{0}
\setcounter{equation}{0}

\def\beq{\begin{equation}}
\def\eeq{\end{equation}}

\appendix

\section{Notation}

The Wigner function is a function $W$ satisfying $\int W = 1$. $W$ does not have to be positive. $|W|$ can be treated as an unnormalized probability distribution. Defining \[\widetilde{W} = \frac{W}{\int |W|},\] we have that $|\widetilde{W}|$ is a normalized probability distribution. Further, let $Z_W = \int |W|$.

For a pure state, the Wigner function also satisfies $2\pi\int W^2 = 1$. We use this identity only as a diagnostic; it is not imposed as a hard constraint during training.

\section{Further Details Regarding Sample-Based Loss Estimation}

As described in the main text, the training of QST-QFlow is based on a sample efficient estimate of the KL divergence. Let $B$ be a bank of samples, where at any time we may pick a distribution $P$, sample an $\alpha$ from $P$, compute $P_\alpha = P(\alpha)$, and append $(\alpha, P_\alpha)$ to $B$. We compute the KL divergence in a sample-efficient manner according to
\begin{align*}
    KL(Q_{exp} || Q_{\theta}) &= \int Q_{exp} \ln \frac{Q_{exp}}{Q_{\theta}}\\
    &= \frac{1}{N}\sum_{i} \int P_i \frac{Q_{exp}}{P_i} \ln \frac{Q_{exp}}{Q_{\theta}}\\
    &\approx \frac{1}{N} \sum_{i} \frac{Q_{exp}(\alpha)}{P_i(\alpha)} \log \frac{Q_{exp}(\alpha)}{Q_{\theta}(\alpha)}\hspace{20pt} \alpha\sim P_i\\
    &= \frac{1}{N} \sum_{(\alpha, P_\alpha) \sim B} \frac{Q_{exp}(\alpha)}{P_\alpha} \log \frac{Q_{exp}(\alpha)}{Q_{\theta}(\alpha)},
\end{align*}
where $Q_{exp}$ is the experimental Q function, $Q_{\theta}$ is QST-QFlow, and $(\alpha, P_\alpha)$ are uniform samples chosen from an arbitrary sample bank $B.$

Similarly, we train QST-WFlow using an $L_1$ loss estimated by importance sampling:
\begin{align*}
    L_1[W_\theta,W] &= \int |W_\theta - W|\\
    &= \frac{1}{N}\sum_{i} \int P_i \frac{\int |W_\theta - W|}{P_i}\\
    &\approx \frac{1}{N}\sum_{i}\frac{|W_\theta(x,p) - W(x,p)|}{P_i(x,p)}\hspace{20pt} (x,p)\sim P_i\\
    &= \frac{1}{N}\sum_{((x,p), P_{(x,p)})\sim B}\frac{|W_\theta(x,p) - W(x,p)|}{P_{(x,p)}},
\end{align*}
where $((x,p), P_{(x,p)})$ are samples chosen uniformly from the sample bank $B$.

\section{Signed-flow Representation of Wigner Functions}

\begin{theorem}
    Any normalized integrable phase-space function $W$ can be expressed as \[W = (\lambda+1)Q_1-\lambda Q_2\] for normalized non-negative probability distributions $Q_1$ and $Q_2$.
\end{theorem}\textbf{Proof}: Let $P_1 = \max(W,0)$ and $P_2 = \max(-W,0)$. By definition, $P_1$ and $P_2$ are non-negative and $\int P_1 = \int P_2 + 1$, since $\int W=1$. Now let \[Q_1 = \frac{P_1}{\int P_1}\] and \[Q_2 = \frac{P_2}{\int P_2},\] and set $\lambda = \int P_2$. We have
\begin{align*}
    (\lambda + 1) Q_1 - \lambda Q_2=&\left(\int P_2 + 1\right) Q_1 - Q_2 \int P_2\\
    =&Q_1\int P_1 - Q_2 \int P_2\\
    =&P_1 - P_2\\
    =&\max(W,0) - \max(-W,0)\\
    =&W,
\end{align*}
as desired.

Note: We may still want a stronger proof that if $W$ is infinitely differentiable we can find $Q_1$ and $Q_2$ that are infinitely differentiable.

When fitting a Wigner function $W_\theta = (\lambda+1)Q_1-\lambda Q_2$, one possible loss function is the $L_1$ loss \begin{align*}
    L_1[W_\theta,W] &= \int |W_\theta - W|\\
    &=\int dx\, dp\, P(x,p)\frac{|W_\theta(x,p) - W(x,p)|}{P(x,p)}\\
    &\approx \frac{1}{N}\sum_{(x,p)\sim P}\frac{|W_\theta(x,p) - W(x,p)|}{P(x,p)}\\
\end{align*} for any probability distribution $P$. Similarly, the (perhaps more important) gradient of the loss can be written as \begin{align*}
    \nabla_\theta L_1[W_\theta,W] &= \int \text{sign}(W_\theta - W) \nabla_\theta W_\theta\\
    &=\int dx\, dp\, P(x,p)\left[\frac{\text{sign}(W_\theta(x,p) - W(x,p))}{P(x,p)}\cdot \nabla_\theta W_\theta(x,p)\right]\\
    &\approx \frac{1}{N}\sum_{(x,p)\sim P}\left[\frac{\text{sign}(W_\theta(x,p) - W(x,p))}{P(x,p)}\cdot \nabla_\theta W_\theta(x,p)\right].
\end{align*} We now propose several options for $P$. \begin{enumerate}
    \item If we choose to set $P = |\widetilde{W}_\theta|$, we get \begin{align*}
        L_1[W_\theta,W] &\approx \frac{1}{N}\sum_{(x,p)\sim |\widetilde{W}_\theta|} \frac{|W_\theta(x,p) - W(x,p)|}{|\widetilde{W}_\theta(x,p)|}\\
        &= \frac{Z_{W_\theta}}{N}\sum_{(x,p)\sim |\widetilde{W}_\theta|} \frac{|W_\theta(x,p) - W(x,p)|}{|W_\theta (x,p)|}\\
        &= \frac{Z_{W_\theta}}{N}\sum_{(x,p)\sim |\widetilde{W}_\theta|} \left|1 - \frac{W(x,p)}{W_\theta (x,p)}\right|
    \end{align*} and \begin{align*}
        \nabla_\theta L_1[W_\theta,W] &\approx \frac{1}{N}\sum_{(x,p)\sim |\widetilde{W}_\theta|}\left[\frac{\text{sign}(W_\theta(x,p) - W(x,p))}{|\widetilde{W}_\theta(x,p)|}\cdot \nabla_\theta W_\theta(x,p)\right]\\
        &= \frac{Z_{W_\theta}}{N}\sum_{(x,p)\sim |\widetilde{W}_\theta|}\left[\frac{\text{sign}(W_\theta(x,p) - W(x,p))}{|W_\theta(x,p)|}\cdot \nabla_\theta W_\theta(x,p)\right].
    \end{align*} 
    
    For this option, we sample from $|W_\theta (x,p)|$ using MCMC methods. Because $Z_{W_\theta}$ factors out and only controls the magnitude of the gradient, we can in principle drop this term. This is especially true because in practice we use special optimizers such as Adam which rescale the gradients anyway. However, if we want to evaluate $Z_{W_\theta}$, we can do so through MCMC methods. We may wish to do so, because $Z_{W_\theta}$ can change as $W_\theta$ changes. It should, however, be bounded by $2\lambda + 1$.

    \item If we choose to set $P = |\widetilde{W}|$, we get \begin{align*}
        L_1[W_\theta,W] &\approx \frac{1}{N}\sum_{(x,p)\sim |\widetilde{W}|} \frac{|W_\theta(x,p) - W(x,p)|}{|\widetilde{W}|}\\
        &= \frac{Z_{W}}{N}\sum_{(x,p)\sim |\widetilde{W}|} \frac{|W_\theta(x,p) - W(x,p)|}{|W(x,p)|}\\
        &= \frac{Z_{W_\theta}}{N}\sum_{(x,p)\sim |\widetilde{W}|} \left|\frac{W_\theta (x,p)}{W(x,p)} - 1\right|
    \end{align*} and \begin{align*}
        \nabla_\theta L_1[W_\theta,W] &\approx \frac{1}{N}\sum_{(x,p)\sim |\widetilde{W}|}\left[\frac{\text{sign}(W_\theta(x,p) - W(x,p))}{|\widetilde{W}(x,p)|}\cdot \nabla_\theta W_\theta(x,p)\right]\\
        &= \frac{Z_{W}}{N}\sum_{(x,p)\sim |\widetilde{W}|}\left[\frac{\text{sign}(W_\theta(x,p) - W(x,p))}{|W(x,p)|}\cdot \nabla_\theta W_\theta(x,p)\right].
    \end{align*} Alternatively, in this case we can just use autodiff directly on $L_1[W_\theta,W]$ to compute $\nabla_\theta L_1[W_\theta,W]$. 
    
    For this option, we sample from $|W(x,p)|$ using MCMC methods. Because $Z_{W}$ factors out and only controls the magnitude of the gradient, we can drop this term. This is especially true because in practice we use special optimizers such as Adam which rescale the gradients anyway. However, if we want to evaluate $Z_{W_\theta}$, we can do so through MCMC methods.

    \item If we choose to set $P = \frac{1}{2}(Q_{1\theta}+Q_{2\theta})$, we get \begin{align*}
        L_1[W_\theta,W] &\approx \frac{1}{N}\sum_{(x,p)\sim \frac{1}{2}(Q_{1\theta}+Q_{2\theta})} \frac{|W_\theta(x,p) - W(x,p)|}{\frac{1}{2}(Q_{1\theta}+Q_{2\theta})}
    \end{align*} and \begin{align*}
        \nabla_\theta L_1[W_\theta,W] &\approx \frac{1}{N}\sum_{(x,p)\sim \frac{1}{2}(Q_{1\theta}+Q_{2\theta})}\left[\frac{\text{sign}(W_\theta(x,p) - W(x,p))}{\frac{1}{2}(Q_{1\theta}+Q_{2\theta})}\cdot \nabla_\theta W_\theta(x,p)\right].
    \end{align*} This method does not require MCMC methods for sampling or evaluating $Z_W.$ However, it has the downside that the samples are not taken from either Wigner function, so each sample may provide less information than in the other proposals. In this case, we could once again differentiate straight through $L_1[W_\theta,W]$ using autodiff and the reparametrization trick.
\end{enumerate}

Another possible loss is the density matrix overlap $1-\text{tr}(\rho(W_\theta)\rho(W))$ where $\rho(\cdot)$ represents the density matrix for a given $W$. In terms of Wigner functions, this can be written as \begin{align*}
    1-\text{tr}(\rho(W_\theta)\rho(W)) &= 1-2\pi \int dx\, dp\, W_\theta(x,p)W(x,p)\\
    &= 1-2\pi \int dx\, dp\, P(x,p)\frac{W_\theta(x,p)W(x,p)}{P(x,p)}\\
    &\approx 1-\frac{2\pi}{N}\sum_{(x,p)\sim P}\frac{W_\theta(x,p)W(x,p)}{P(x,p)}.
\end{align*} The gradient of this loss can be written as \begin{align*}
    \nabla_\theta(1-\text{tr}(\rho(W_\theta)\rho(W))) &= -2\pi \int dx\, dp\, W(x,p) \nabla_\theta W_\theta(x,p)\\
    &= -2\pi \int dx\, dp\, P(x,p)\frac{W(x,p)}{P(x,p)}\cdot \nabla_\theta W_\theta(x,p)\\
    &\approx -\frac{2\pi}{N}\sum_{(x,p)\sim P}\frac{W(x,p)}{P(x,p)}\cdot \nabla_\theta W_\theta(x,p).
\end{align*} We now propose several options for $P$. \begin{enumerate}
    \item If we choose to set $P = |\widetilde{W}_\theta|$, we get \begin{align*}
        1-\text{tr}(\rho(W_\theta)\rho(W)) &\approx 1-\frac{2\pi}{N}\sum_{(x,p)\sim |\widetilde{W}_\theta|}\frac{W_\theta(x,p)W(x,p)}{|\widetilde{W}_\theta(x,p)|}\\
        &= 1-\frac{2\pi Z_{W_\theta}}{N}\sum_{(x,p)\sim |\widetilde{W}_\theta|} W(x,p)\cdot \text{sign}(W_\theta(x,p))
    \end{align*} and \begin{align*}
        \nabla_\theta (1-\text{tr}(\rho(W_\theta)\rho(W))) &\approx -\frac{2\pi}{N}\sum_{(x,p)\sim |\widetilde{W}_\theta|}\frac{W(x,p)}{|\widetilde{W}_\theta(x,p)|}\cdot \nabla_\theta W_\theta(x,p)\\
        &= -\frac{2\pi Z_{W_\theta}}{N}\sum_{(x,p)\sim |\widetilde{W}_\theta|}\frac{W(x,p)}{|W_\theta(x,p)|}\cdot \nabla_\theta W_\theta(x,p).
    \end{align*} 

    \item If we choose to set $P = |\widetilde{W}|$, we get \begin{align*}
        1-\text{tr}(\rho(W_\theta)\rho(W)) &\approx 1-\frac{2\pi}{N}\sum_{(x,p)\sim |\widetilde{W}|}\frac{W_\theta(x,p)W(x,p)}{|\widetilde{W}(x,p)|}\\
        &= 1-\frac{2\pi Z_W}{N}\sum_{(x,p)\sim |\widetilde{W}|} W_\theta(x,p)\cdot \text{sign}(W(x,p))
    \end{align*} and \begin{align*}
        \nabla_\theta (1-\text{tr}(\rho(W_\theta)\rho(W))) &\approx -\frac{2\pi}{N}\sum_{(x,p)\sim |\widetilde{W}|}\frac{W(x,p)}{|\widetilde{W}(x,p)|}\cdot \nabla_\theta W_\theta(x,p)\\
        &= -\frac{2\pi Z_W}{N}\sum_{(x,p)\sim |\widetilde{W}|} \text{sign}(W_\theta(x,p))\nabla_\theta W_\theta(x,p).
    \end{align*}  Alternatively, in this case we can just use autodiff directly on $1-\text{tr}(\rho(W_\theta)\rho(W))$ to compute $\nabla_\theta (1-\text{tr}(\rho(W_\theta)\rho(W)))$.

    \item If we choose to set $P = \frac{1}{2}(Q_{1\theta}+Q_{2\theta})$, we get \begin{align*}
        1-\text{tr}(\rho(W_\theta)\rho(W)) &\approx 1-\frac{2\pi}{N}\sum_{(x,p)\sim \frac{1}{2}(Q_{1\theta}+Q_{2\theta})}\frac{W_\theta(x,p)W(x,p)}{\frac{1}{2}(Q_{1\theta}+Q_{2\theta})}
    \end{align*} and \begin{align*}
        \nabla_\theta (1-\text{tr}(\rho(W_\theta)\rho(W))) &\approx -\frac{2\pi}{N}\sum_{(x,p)\sim \frac{1}{2}(Q_{1\theta}+Q_{2\theta})}\frac{W(x,p)}{\frac{1}{2}(Q_{1\theta}+Q_{2\theta})}\cdot \nabla_\theta W_\theta(x,p).
    \end{align*}  This method does not require MCMC methods for sampling or evaluating $Z_W$. However, it has the downside that the samples are not taken from either Wigner function, so each sample may provide less information than in the other proposals. In this case, we could once again differentiate straight through $\text{tr}(\rho(W_\theta)\rho(W))$ using autodiff and the reparametrization trick.
\end{enumerate}

\begin{figure*}
    \centering
    \includegraphics[width=\linewidth]{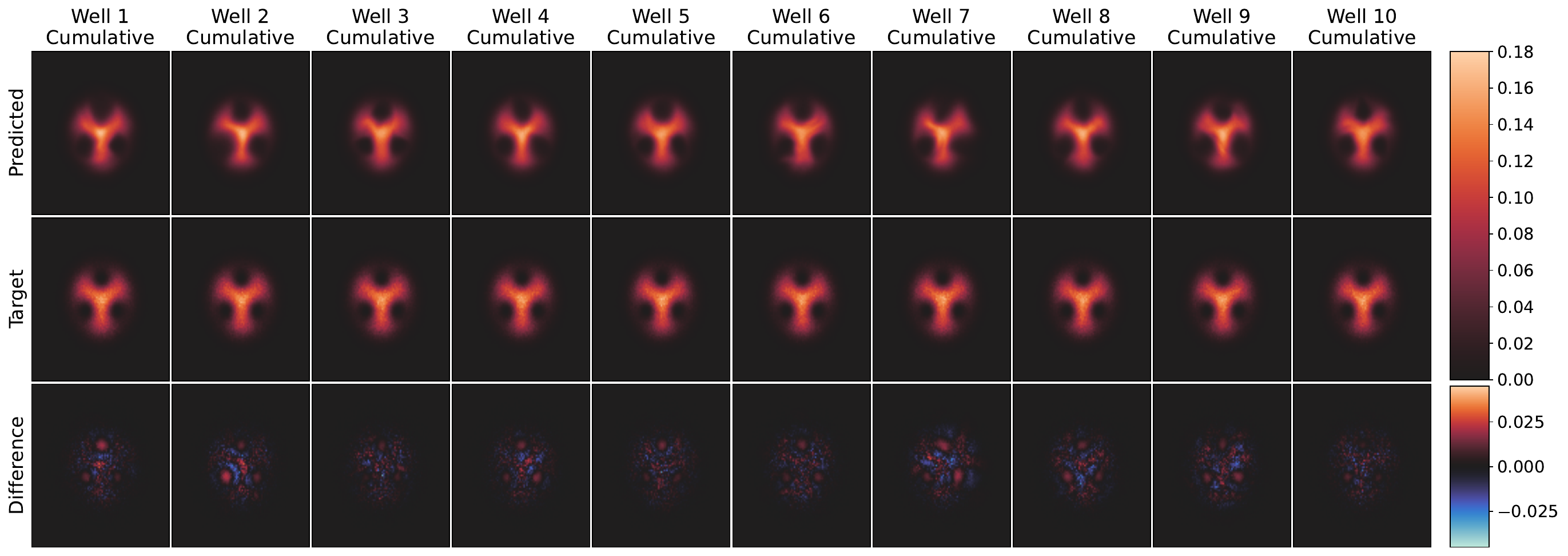}
    \caption{$Q$-function tomography of a 10-mode continuous-variable quantum state.}
    \label{fig:10_well}
\end{figure*}

\subsection{Definitions of Quantum States}\label{app: state definitions}

Throughout this paper, we use the notation $\ket{\cdot}_\alpha$ to denote that the ket is parametrized in terms of coherent states and the notation $\ket{\cdot}_n$ to denote that the ket is parametrized in terms of Fock states.

Now, we make the following definitions:

\begin{align*}
    \ket{\text{cat}(\alpha)} &= \frac{1}{\sqrt{2+2e^{-2|\alpha|^2}}}\left(\ket{\alpha}_\alpha - \ket{-\alpha}_\alpha\right)\\
    \ket{\text{num}(s)} &= \sqrt{\frac{8-(-1)^{s}-\sqrt{17}}{6}}\ket{s}_n + \sqrt{\frac{\sqrt{17}-2+(-1)^{s}}{6}}\ket{3+s}_n\\
    \ket{\text{binom}(N,S)} &= \sum_{m=0}^{N+1} \sqrt{\frac{\binom{N+1}{m}}{2^{N+1}}}\ket{(S+1)m}_n\\
    \ket{\text{GKP}(\delta, C)} &= \sum_{-C\le x,y\le C} e^{-\left[-\frac{\pi}{2}\left(\delta^2 [4x^2+y^2] -i 2xy\right)\right]}\left|(2x+iy)\sqrt{\frac{\pi}{2}}\right\rangle_\alpha.
\end{align*}

\section{More Experimental Details} \label{app: exp}
\paragraph{Common Experimental Details:} In this section, we lay out the default setup for all experiments. All defaults presented here hold for each experiment in this paper unless specifically indicated in the relevant section below.

All experiments are implemented in Python using libraries such as Jax \cite{jax2018github}, Diffrax \cite{diffrax}, Equinox \cite{kidger2021equinox}, Optax \cite{deepmind2020jax}, and Haiku \cite{haiku2020github}.

We use a Continuous Normalizing Flow \cite{grathwohl2018ffjordfreeformcontinuousdynamics} for all experiments. In particular, we represent each flow as a mapping $f_\theta(x,t)$ and define samples $X$ from the flow according to the process\begin{align*}
    X = x(t=0.5)\,\odot e^{-s} \odot r, \quad \quad \frac{dx}{dt} = f_\theta(x,t), \quad \quad x(t=0) \sim \mathcal{N}(0,1),
\end{align*}
where $\odot$ denotes the elementwise product of two vectors.
In the above, $s$ is a trainable vector parameter (part of $\theta$) and $r$ is a non-trainable vector parameter provided at initialization (we initialize it to a sample-based estimate of the standard deviation of the target distribution, as discussed below).
We can similarly represent the probability distribution represented by the flow as the solution to the ordinary differential equations
\begin{align*}
    p_\theta(X) &= p_\theta(X\, e^s/r,t=0.5)e^{\sum_i s_i - \log r_i},\\
    \frac{d\log p_\theta(x,t)}{dt} &= -\text{Tr}\left[\frac{\partial f_\theta(x,t)}{\partial x}\right]_{x=x(t)},\\
    p_\theta(x,t=0) &= (2\pi)^{-\text{dim}(x)/2}\exp\left[-\frac{|x|^2}{2}\right].
\end{align*}

Similar to FFJORD \cite{grathwohl2018ffjordfreeformcontinuousdynamics}, we parametrize $f_\theta$ as follows:
\begin{align*}
    f_\theta(x,t) &= h_L(x,t),\\
    h_i(x,t) &= \xi\left(\left(W_i^{(1)} h_{i-1} + b_i^{(1)}\right)\odot \sigma\left(w_i^{(2)}t + b_i^{(2)}\right) + w_i^{(3)}t\right) + h_{i-1},\\
    h_0 &= x,
\end{align*}
where $\sigma$ represents the sigmoid function, $\xi$ is another activation, and $L$ is the network depth. The dimensions of the objects above are
\begin{align*}
    W_1^{(1)} &\in \mathbb{R}^{N\times N_{in}},\\
    W_i^{(1)} &\in \mathbb{R}^{N\times N}, \quad\quad 1<i<L\\
    W_L^{(1)} &\in \mathbb{R}^{N_{in}\times N},\\
    w_i^{(2)},w_i^{(3)},b_i^{(1)},b_i^{(2)} &\in \mathbb{R}^{N},\quad\quad i<L\\
    w_L^{(2)},w_L^{(3)},b_L^{(1)},b_L^{(2)} &\in \mathbb{R},
\end{align*}
where $N_{in}$ is the dimension of $x$ and $N$ is the hidden dimension. Depending on the task, we set $\xi(x)$ to one of $\sin(30x)$ or $\text{GELU}(x)$. We initialize the weights according to
\begin{align*}
    w_1^{(1)}&\sim \begin{cases}
        \mathcal{U}\left(-\frac{1}{30}\sqrt{\frac{6}{N_{in}}},\; \frac{1}{30}\sqrt{\frac{6}{N_{in}}}\right) &\text{if }\xi(x) =\sin(30x)\\
        T\left[\mathcal{N}\left(-\frac{1}{\sqrt{N_{in}}},\frac{1}{\sqrt{N_{in}}}\right)\right]&\text{if }\xi(x) =\text{GELU}(x),\\
    \end{cases}\\
    w_i^{(1)}&\sim \begin{cases}
        \mathcal{U}\left(-\frac{1}{N},\; \frac{1}{N}\right) &\text{if }\xi(x) =\sin(30x)\\
        T\left[\mathcal{N}\left(-\frac{1}{\sqrt{N}},\frac{1}{\sqrt{N}}\right)\right] &\text{if }\xi(x) =\text{GELU}(x),\\
    \end{cases}\quad\quad\quad\quad \quad\quad\quad\quad i>1\\
    w_i^{(2)}&\sim T\left[\mathcal{N}\left(0,1\right)\right],\\
    w_i^{(3)}&\sim T\left[\mathcal{N}\left(0,1\right)\right],\\
    b_i^{(1)}&= 0,\\
    b_i^{(2)}&= 0,
\end{align*}
where $T$ represents the truncation of the random variable to the range $[-2,2]$.

By default, we use $\xi(x)=\sin(30x)$, $N=30$, and $L=5$. We solve the forward and reverse ODEs using Tsitouras' 5/4 method with an adaptive step size controlled by a PID controller with default step size 0.1 and both relative and absolute tolerance set to $10^{-4}$. We backpropagate through the solver to compute gradients.

We use the learning rate schedule
\begin{equation*}
    \text{lr} = \begin{cases}
        \frac{E\text{lr}_{max}}{wE_{max}} & \text{if } E < wE_{max}\\
        \frac{1}{2}\left[1+\cos\left(2\pi\frac{E - wE_{max}}{(1-w)E_{max}}\right)\right] & \text{if } wE_{max}\le E \le E_{max},
    \end{cases}
\end{equation*}
where $E_{max}$ is the total number of steps, $E$ is the current step, and $w$ is the percentage of steps that are used for warmup. By default, $E_{max} = 500$ and $w=0.1$.

We train using AdamW \cite{adamw} with default learning rate $4\cdot 10^{-3}$, betas $(\beta_1,\beta_2) = (0.9, 0.999)$, and a weight decay of 0.0001. We also clip gradients larger than 100.

We add 100 new samples from the current model to the sample bank every 10 epochs. For $Q$ functions, we sample directly from $Q_\theta$. For Wigner functions with $W_\theta = (\lambda+1) Q_\theta^{(1)} - \lambda Q_\theta^{(2)}$, we use the positive proposal $\frac{1}{2}\left(Q_\theta^{(1)} + Q_\theta^{(2)}\right)$. For each training step, we select $N_s$ samples uniformly from the sample bank for training using the losses described in the main text.

\textbf{Figure 2 Experimental Details:} 
This experiment closely follows the default parameters, with the following exceptions:
\begin{itemize}
    \item For the Q-Function GKP state and the W-Function Binomial State, we set $E_{max}=1000.$
    \item For the W-Function GKP state and Fock state, we set $E_{max}=4000.$
\end{itemize}

\textbf{Figure 3 Experimental Details:} This experiment closely follows the default parameters, with the following exceptions:
\begin{itemize}
    \item We do not use the sample bank. Instead we train the model on every point in the training dataset (which is inside the circle described in \cite{QST_CGAN}) each training step. We use the L1  between the sample values and the $W_\theta$ predictions for training.
    \item We use $E_{max}=1000.$
    \item We use $\xi(x) = \text{GELU}(x).$
\end{itemize}

\textbf{Figure 4 Experimental Details:} This experiment closely follows the default parameters, with the following exceptions:
\begin{itemize}
    \item We use $E_{max}=8000$. 
    \item We use $N_s=600.$
    \item We use $\xi(x) = \text{GELU}(x).$
    \item We use $N=80.$
\end{itemize}

\textbf{Figure S1 Experimental Details:} This experiment closely follows the default parameters, with the following exceptions:
\begin{itemize}
    \item We use $E_{max}=8000$. 
    \item We use $N_s=500.$
    \item We use $\xi(x) = \text{GELU}(x).$
    \item We use $N=80.$
\end{itemize}

\end{document}